\documentclass[reqno,12pt]{amsart}

\usepackage{enumerate}
\usepackage[margin=1in]{geometry}
\usepackage{enumitem}
\usepackage{ifpdf}
\usepackage{amsmath}
\usepackage{amsfonts}
\usepackage{amssymb}
\usepackage{amsaddr}
\usepackage{amsthm}
\usepackage{amsrefs}
\usepackage{mathrsfs}
\usepackage{cite}
\usepackage[hyperfootnotes=false]{hyperref}
\usepackage[dotinlabels]{titletoc}
\usepackage{nicefrac}
\usepackage{float}
\usepackage[multiple]{footmisc}
\usepackage{graphicx}
\usepackage{dcolumn}
\usepackage{tabu}

\newcommand{\norm}[1]{\left\lVert {#1} \right\rVert}
\newcommand{\ket}[1]{\ensuremath{\left|#1\right\rangle}}

\newtheorem{thm}{Theorem}[section]

\theoremstyle{definition}

\theoremstyle{remark}

\theoremstyle{note}

\newcolumntype{A}{D{.}{.}{2.3}}

\makeatletter

\newcommand{\Rmnum}[1]{\expandafter\@slowromancap\romannumeral #1@}
\makeatother
\author{Asif Shakeel}
\address{CA 92093}
\email{asif.shakeel@gmail.com}
\title[Neighborhood-History QW] {Neighborhood-History Quantum Walk}

\begin{document}

\begin{abstract}
History dependent discrete time quantum walks (QWs) are  often studied for their lattice traversal properties. A particular  model in the literature uses the  state  of a   memory  qubit at each site to record visits and to control the dynamics of the walk. We generalize this model to the  {\it neighborhood-history quantum walk} (NHQW), in which  the walk dynamics and  the state of the memory qubits in a {\it neighborhood} of the particle's position are interdependent. To demonstrate it, we construct an NHQW on a one-dimensional lattice, with a simple neighborhood.  Several dynamically interesting history dependent QWs   can be realized as single-particle sectors of quantum lattice gas automata (QLGA).  In contrast, the   NHQW constructed in this paper is realized as a single-particle sector of the more general quantum cellular automaton (QCA).  The   complexity of the NHQW  dynamics presents a promising avenue toward richer walk strategies and a potentially useful model of QWs for the %%%%% FIXED 
Noisy Intermediate-Scale Quantum (NISQ) %%%%
era of quantum computing.
%%%ADDED
It also modifies QWs to conceivably allow for modeling fundamental physics  incorporating  quantum field interactions with particles.
\end{abstract}

\maketitle
  
\section{Introduction} \label{section:intro}
 Discrete time quantum walks (QWs) have an established status as models of   quantum   search algorithms,~\cite{mnrs:svqw,kmor:fiaeadqw,sm:qsmcba,bib:skw,bib:Montanaro20150301,pm:gqn, pmcm:qgcn}. Being   capable of  universal quantum computation~\cite{lcetk:uqcdqrw}, they are contenders  for quantum computing tasks as well.  It is noteworthy  that physical realizations of QWs have emerged,  and as they  develop further~\cite{plp:rqw,qlm:eqwqp}, they are useful for  demonstrating and manipulating quantum behavior.  Altogether,  strong interest  continues to grow in understanding and designing  their  lattice traversal properties~\cite{bca:qwdmc,bca:qctrw,mg:odqwwm,rbg:qwwmrccf,pbhmpbk:nrnrqw,crt:qwwtsa1d,bib:Rosmanis,bib:akrcmqwf} to achieve faster and  more accurate search algorithms.

 From a search standpoint, to perform  a more efficient search with a QW, a desirable feature may be {\it self-avoidance}: reduced probability of visiting previously visited sites in favor of unvisited sites~\cite{pbhmpbk:nrnrqw,crt:qwwtsa1d}.   With this goal, a  model of  QW proposed by Camilleri {\it et al.}~\cite{crt:qwwtsa1d}  expands the QW model to append a memory qubit to each lattice site, to maintain  a record of particle visits.  When the  particle hops to a  site,  the interaction between the site's memory  qubit and the particle modifies the memory qubit state and controls the particle direction for the next hop. We would like, instead,  to make the  particle interact  with the  neighboring  memory qubits prior to the hop.   That requires a further  expansion of the QW  model to bring the neighboring memories in   interaction with the particle. 
 
 As quantum computers transition to the  Noisy Intermediate Stage Quantum (NISQ)~\cite{p:qcnisq} regime, QWs would be realizable as quantum circuits. This model, then, may have significance in modeling QWs on NISQ computers,  modeling the persistent coupling between neighboring qubits. %%%%%%%%%% ADDED %%%%%%%%%%%%%%%%%
In the current state of development of quantum computers,  qubits in NISQ machines, depending on the physical medium and the gate implementation, do not have complete connectivity among them. Optimizations are being proposed for  compilation of algorithmic circuits to gates~\cite{zpw:emmqciqa}. Generally,  a subset of qubits control another subset,  typically  as a directed partial graph among the qubits. This leads to algorithmic operations involving neighboring circuit qubits to necessarily be implemented through a sequence of intermediate operations with other neighboring qubits (swaps). Also, typically not all the qubits of a machine are being used for a circuit, at least  not \textit{actively}  in a given algorithmic step. Modeling a quantum walk implementation on a NISQ machine would  lend itself to a NHQW model to realistically obtain the dynamics in a noisy environment, in which coherence times and couplings with \textit{inactive} qubits play a role.
%%%%%%%%%%%%%%%%%%%%%%%%%%%%%%%%%%%%%%%%%%

An enduring topic of investigation since Feynman's original work~\cite{f:spwc} is simulation of fundamental physics using quantum computers~\cite{dm:qmlga}. Recently, work in  simulating Dirac equation on triangular and honeycomb lattices through QWs~\cite{adma:deqwhctl}, as well as  QWs to  simulate physics of  Dirac fermions in curved spacetime~\cite{mbd:qwmdfcst},   extends these ideas towards quantum gravity. Interactions of quantum fields with particles in a true quantum mechanical sense would require incorporating information about the fields and their  interactions in the neighborhood of particle. NHQW could serve as a useful model in this regard.

This paper is organized as follows. In Section~\ref{sec:QW} we give a brief description of the basic quantum walk (QW). In Section~\ref{sec:ndqw} we introduce the {\it neighborhood-history quantum walk} (NHQW), and as an example, recall the model from~\cite{crt:qwwtsa1d} that corresponds to a trivial neighborhood. In Section~\ref{sec:ndqwlrsym} we consider  NHQWs over  a  simple neighborhood (left/right symmetric). First we discuss a model of scattering similar to the one in~\cite{crt:qwwtsa1d} for   this neighborhood, and give a plausibility argument for why it would behave similarly. Then we move on to the NHQW scheme we are interested in, fundamentally different from this model, embodying internal (spin) parameters and external interaction parameters.  This construction is related to the work in~\cite{lsw:lduobag} on unentangled orthogonal bases (UOB). We show  simulation results of the walk patterns obtained by choosing different sets of parameters. Adjusting these parameters leads to  a range of  walk behaviors. Next, we describe a quantum cellular automaton (QCA) whose restriction to a  single-particle sector is this NHQW. It turns out that this QCA is a concatenation of QLGA, and such QCA are investigated in~\cite{bib:msqcawp}. That paper studies  QCA that have no particle description at the scale at which the dynamics are homogeneous. Section~\ref{sec:conc} concludes with some  observations and future directions.

\section{Quantum walk} \label{sec:QW}
A QW consists of a  quantum particle hopping  on a lattice from site to site and scattering on arrival at the sites. Let us understand a basic QW on a   one-dimensional finite sized lattice\footnote{Generalization to infinite lattice and multiple dimensions can be done through the framework in~\cite{bib:slwqcaqlga,bib:sdlhdQWqlga}.} of length $N$.  We denote the lattice $\mathcal{L}=\mathbb{Z}_N$.  The particle has a position on the lattice, $x \in \mathcal{L} = \mathbb{Z}_N$,  and a velocity,   $v \in \{ +1,-1\}$,  giving the   direction of its next hop. The state of the particle is a unit vector in a Hilbert space with its basis elements labeled jointly by position and velocity, i.e., the QW  Hilbert space is a tensor product of  Hilbert spaces associated with position and velocity.   The  position Hilbert space, $\mathcal{H}_\mathcal{L} = \mathbb{C}^N$, is $N$-dimensional and has basis  $\{\ket{x}:  x \in \mathcal{L} = \mathbb{Z}_N\}$.\footnote{$\mathbb{Z}_N = \mathbb{Z}/(N)$.} The  velocity Hilbert space is two-dimensional,  $\mathcal{H}_V = \mathbb{C}^2$, with basis  $\{\ket{v} : v = +1,-1\}$. The $\ket{+1}$ velocity vector   will be referred to as the {\it right-moving} and   $\ket{-1}$ velocity vector  as the {\it left-moving} on the lattice.      Thus, the QW Hilbert space is%$\mathcal{H}$ is 
\begin{equation*}
\mathcal{H} = \mathcal{H}_\mathcal{L}\otimes \mathcal{H}_V = \mathbb{C}^N \otimes\mathbb{C}^2.
\end{equation*}
State of the QW is a vector $\psi \in \mathcal{H}$ of unit norm $\norm{\psi} = 1$.

At each time step, the state transitions in two successive stages:
 \begin{enumerate}[label=(\roman{*})] 
\item \label{scat}  {\it Scattering} (or {\it coin toss})
\begin{equation*}
S = I\otimes U_S : \ket{x} \ket{v} \mapsto \ket{x} U_S(\ket{v}),
\end{equation*}
where $U_S$ is a unitary map on the velocity space $\mathbb{C}^2$. It is also commonly called a ``coin" operator.
\item \label{prop} {\it Advection} (or {\it shift})
\begin{equation*}
A : \ket{x}   \ket{v} \mapsto \ket{x+v}   \ket{v}.
\end{equation*}
Here, the addition is modulo $N$ to allow ``wrap-around" at the edges of the lattice.
\end{enumerate}   
The overall QW transition rule is denoted $T$, a composition of scattering followed by advection, 
\begin{equation*}
T = A S = A (I\otimes U_S).
\end{equation*}

\section{Neighborhood-history  quantum walk}  \label{sec:ndqw}
We  define the neighborhood-history quantum walk (NHQW).  The model is described in somewhat general terms. Denote by  $\mathcal{L} = \mathbb{Z}_{N_1} \times \ldots \times \mathbb{Z}_{N_n}$ the lattice on which the particle walks, so that each position basis element $\ket{x}$ is labeled by  $x \in \mathcal{L}$.  A velocity basis element $\ket{v}$ is labeled from some finite subset of the lattice: $v \in \mathcal{V} \subset \mathcal{L}$.

The Hilbert space of an NHQW is 
\begin{equation*} %\label{eq:hsndqw}
\mathcal{H} =\mathcal{H}_\mathcal{L} \otimes \mathcal{H}_V \otimes  \bigotimes_{k \in \mathcal{L}} \mathcal{H}_{M_k} =\mathbb{C}^{\left\vert  \mathcal{L} \right\vert} \otimes\mathbb{C}^{\left\vert  \mathcal{V} \right\vert}\otimes  \bigotimes_{k \in \mathcal{L}}  \mathbb{C}^2.
\end{equation*}
%The order  of factors respectively give position, velocity, and memory elements of the walk.  
$\mathcal{H}_\mathcal{L} = \mathbb{C}^{\left\vert  \mathcal{L} \right\vert}, \mathcal{H}_V = \mathbb{C}^{\left\vert  \mathcal{V} \right\vert}, \mathcal{H}_{M_k} = \mathbb{C}^2$,  correspond to Hilbert spaces for position, velocity, and memory qubit at position $k\in \mathcal{L}$, respectively. 
A basis element of this Hilbert space is given as
\begin{equation*} %\label{eq:fhsn}
\ket{x} \ket{v}  \bigotimes_{k \in \mathcal{L}}  \ket{ m_k},
\end{equation*}
with $x \in \mathcal{L}$, $v \in \mathcal{V}$, and $ m_k \in \{0,1\}$.

In this model, we allow a set of memory elements from the sites surrounding every site to participate in the scattering.  This is given by the {\it neighborhood} denoted by  a finite subset of the lattice,  $\mathcal{M} \subset \mathcal{L}$.  The neighborhood of site $x$ is just the  neighborhood shifted to $x$: $\mathcal{M}_x = x+ \mathcal{M}$. %Let $\bar{\mathcal{M}}_x = \mathbb{Z}_N \setminus \mathcal{M}_x$.

An NHQW transition has the familiar steps of scattering (implicitly controlled by $\ket{x}$ through the neighborhood selection) and advection:
\begin{enumerate}[label=(\roman{*})] 
\item \label{memo3}  Scattering,
\begin{equation*}
 S :   \ket{x} \ket{v}  \bigotimes_{k \in \mathbb{Z}_N} \ket{ m_k} \mapsto \ket{x}    \bigotimes_{l \in \mathcal{L} \setminus \mathcal{M}_x} \ket{ m_l}U_S\bigg( \ket{v} \bigotimes_{k \in \mathcal{M}_x} \ket{ m_k}\bigg),
\end{equation*}
%%%%%%%%%%%%************************************************************%%%
%%%%%%%%%%%%*************** FIXME: HILBERT SPACE FOR U_S *********%%%%%%%%%
%%%%%%%%%%%%************************************************************%%%
where $U_S$ is the {\it neighborhood scattering operator}  acting unitarily  on velocity and neighborhood memory Hilbert space: $\mathcal{H}_V \otimes  \bigotimes_{k \in \mathcal{M}_x} \mathcal{H}_{M_k}$

% the velocity $\ket{v}$ and the neighborhood memory qubits $\ket{ m_k}$ for $k \in \mathcal{M}_x$.

\item \label{prop3} Advection, 
\begin{equation*}
A :   \ket{x} \ket{v}  \bigotimes_{k \in \mathbb{Z}_N} \ket{ m_k} \mapsto \ket{x+v} \ket{v} \bigotimes_{k \in \mathbb{Z}_N} \ket{ m_k}.
\end{equation*}

\end{enumerate}    
The NHQW transition is given as 
\begin{equation*}
T =  A    S.
\end{equation*}

In this paper, all the examples will be concerned with 
 the lattice $\mathcal{L}=\mathbb{Z}_N$, for which the NHQW Hilbert space  is
\begin{equation}\label{eq:fqsq}
\mathcal{H} = \mathbb{C}^N \otimes \mathbb{C}^2 \otimes \bigotimes_{k \in \mathbb{Z}_N} \mathbb{C}^2,
\end{equation}
with $\mathcal{H}_\mathcal{L} = \mathbb{C}^N, \mathcal{H}_V = \mathbb{C}^2, \mathcal{H}_{M_k} = \mathbb{C}^2$, specifying Hilbert spaces corresponding to position, velocity, and memory qubit at position $k\in \mathbb{Z}_N$, respectively. We  write a computational basis element of the  Hilbert space $\mathcal{H}$ as
\begin{equation} \label{eq:shbasis}
 \ket{x} \ket{v}  \ket{m_0  \ldots m_{N-1}},
\end{equation}
where $m_k \in \{0,1\}$  denotes the memory qubit  corresponding to site $k$.

In~\cite{crt:qwwtsa1d}, we are introduced to the notion of using memory qubits, one per lattice site,  to record the history of visits, with $\ket{0}$  denoting the ``not visited" state and $\ket{1}$  denoting the ``have visited" state. Scattering happens in two stages. One stage updates the memory qubit on particle arrival at the site. The other stage  performs a memory-controlled operation on the velocity. These operations are implicitly controlled by $\ket{x}$ which determines the memory location involved in the scattering.\footnote{This description  is the one in~\cite{bib:sdlhdQWqlga} adapted from the original in~\cite{crt:qwwtsa1d}.} 
\begin{enumerate}[label=(\roman{*})] 
\item \label{memo3}  Memory,  
\begin{equation*}
M :   \ket{x} \ket{v}  \ket{m_0  \ldots m_{N-1}} \mapsto \ket{x}   \ket{v}  \ket{m_0 \ldots m_{x-1}}U_M(\ket{m_x}) \ket{m_{x+1}\ldots m_{N-1}},
\end{equation*}
where $U_M$ is a symmetric  matrix,\footnote{All the matrices are assumed to be with respect to the specified  bases for the respective Hilbert spaces.} parameterized by $\theta_m$ (the memory strength),
\begin{equation*} %\label{memun}
 U_M=\begin{pmatrix}
 \cos\theta_m&i \sin\theta_m\\
i \sin\theta_m&  \cos\theta_m\\
\end{pmatrix}.
\end{equation*}
\item \label{scat2}  Ricochet, %a controlled operation on the $\ket{v}$ factor,   controlled by the qubit $\ket{m_x}$,  
\begin{equation*}
R :   \ket{x}  \ket{v}  \ket{m_0 \ldots m_x \ldots m_{N-1}}  \mapsto \ket{x}   R_{m_x}(\ket{v})  \ket{m_0 \ldots m_x \ldots m_{N-1}},
\end{equation*}
The state of the memory, $m_x$, determines the scattering $R_{m_x}$. For $m_x = 0$ (unvisited state), it is
\begin{equation*} %\label{fortun}
 R_{0}=\frac{1}{\sqrt{2}}\begin{pmatrix}
 1&i \\
i &  1\\
\end{pmatrix},
\end{equation*}
and  for  $m_x = 1$ (visited), it is the symmetric scattering matrix, in turn parameterized by $\theta_b$ (the ``back action''),
\begin{equation*} %\label{bactun}
 R_{1}=\begin{pmatrix}
 \cos\theta_b&i \sin\theta_b\\
i \sin\theta_b&  \cos\theta_b\\
\end{pmatrix}.
\end{equation*} 
\end{enumerate}    
Scattering  is the composition of the above operations
\begin{equation*}
S =     R M.
\end{equation*}
Advection acts as for the normal QW, 
\begin{equation*}
A : \ket{x}   \ket{v}  \ket{m_0 \ldots m_{N-1}}  \mapsto \ket{x+v}   \ket{v} \ket{m_0 \ldots m_{N-1}}.
\end{equation*}
%Since $R$ and $M$ are non-identity on every factor, we use the term ``generalized scattering'' denoted $ S$, to be  %in~\cite{crt:qwwtsa1d}
%\begin{equation}  \label{gsdef}
% S = R M.
%\end{equation} 
Overall transition is
\begin{equation*}
T =  A    S.
\end{equation*}

In~\cite{crt:qwwtsa1d},  numerical simulations of this model for  several instructive pairs of values of the parameters  $\theta_m$ and $\theta_b$  are performed and discussed.  
In the same vein, as a contrast to the main example of this paper that will follow, we first describe a memory-ricochet model of NHQW with  a non-trivial neighborhood, to incrementally move away from the trivial neighborhood in the above model.

\section{NHQW  with left/right symmetric neighborhood} \label{sec:ndqwlrsym}
We take  the lattice and Hilbert space as  given in eq.~\eqref{eq:fqsq}, and add the neighborhood %with the basis elements in eq.~\eqref{eq:shbasis}.  %We take the basic QW lattice $\mathbb{Z}_N$, and velocity basis  $\{\ket{v} : v \in \{+1,-1\} \}$.  
% \begin{equation*}\label{eq:ndhs}
%\mathcal{H} =\mathcal{H}_\mathcal{L} \otimes \mathcal{H}_V \otimes  \bigotimes_{k \in \mathbb{Z}_N} \mathcal{H}_{M_k} = \mathbb{C}^N \otimes \mathbb{C}^2 \otimes \bigotimes_{k \in \mathbb{Z}_N} \mathbb{C}^2,
%\end{equation*}
%\begin{equation*} \label{eq:ndbasis}
% \ket{x} \ket{v}  \ket{m_0  \ldots m_{N-1}}.
%\end{equation*}
in which the  memory at the current site is not involved in the scattering, but the left/right neighbors are. This is   the  left/right symmetric neighborhood  $\mathcal{M}=\{-1,+1\}$.    If the current  particle position is $x$, the neighboring  memory locations are $\mathcal{M}_x = \{x-1,   x+1\}$.
 \subsection{A memory-ricochet NHQW}  \label{subsec:memric}
 The basic form of the memory-ricochet scattering for this neighborhood is as before. % In the memory stage, the neighborhood  memories undergo a parametrized unitary transformation. %,   controlled by the velocity (a generalization of the SHQW  memory-ricochet model). In the ricochet stage the velocity is transformed by a unitary transformation  controlled by the  neighboring memory qubits.
% controlled by $\ket{v}$, recording the particle's visit in the state of the neighboring memories $\ket{m_{x-1} m_{x+1}}$,
% a unitary operation on the $\ket{v}$ factor,   controlled by the neighboring memory qubits $\ket{m_{x-1} m_{x+1}}$.
The operations involved in the scattering are,
\begin{enumerate}[label=(\roman{*})] 
\item \label{memo4}  Memory,  
\begin{align*}
	M :   \ket{x} \ket{v}  \ket{m_0  \ldots m_{N-1}} \mapsto \ket{x}&   \ket{v}  \ket{m_x} \ket{m_0} \ldots  \\		& \ket{m_{x-1}} U_v(\ket{m_{x-1}} \ket{m_{x+1}}) \ket{m_{x+1}\ldots m_{N-1}},
\end{align*}
where $U_v$  is a symmetric  matrix parameterized by  $\theta_v$, the memory strength, %(controlled by $\ket{v}$)
\begin{equation*} %\label{memun2}
 U_v=\begin{pmatrix}
 \cos\theta_v&i \sin\theta_v\\
i \sin\theta_v&  \cos\theta_v\\
\end{pmatrix}.
\end{equation*}
\item \label{scat3}  Ricochet,  

Denoting  $\mathbf{m}_x = m_{x-1} m_{x+1}$, ricochet is 
\begin{equation*}
R :   \ket{x}  \ket{v}  \ket{m_0 \ldots m_x \ldots m_{N-1}}  \mapsto \ket{x}   U_{\mathbf{m}_x}(\ket{v})  \ket{m_0 \ldots m_x \ldots m_{N-1}},
\end{equation*}
where $U_{\mathbf{m}}$, $\mathbf{m} \in \{00,01,10,11\}$, are parameterized as
\begin{equation*} %\label{fortun2}
 U_{\mathbf{m}}=\begin{pmatrix}
 \cos\theta_{\mathbf{m}}&i \sin\theta_{\mathbf{m}}\\
i \sin\theta_{\mathbf{m}}&  \cos\theta_{\mathbf{m}}\\
\end{pmatrix},
\end{equation*}
and  the parameter $\theta_{\mathbf{m}}$ is the back-action parameter corresponding to ${\mathbf{m}}$. %depends on  corresponds to the neighborhood memory state $\ket{{\mathbf{m}_x}}$.
\end{enumerate} 
In terms of  $R$ and $M$  the  scattering is 
\begin{equation*}  %\label{gsdef2}
 S = R M.
\end{equation*} 
Advection, $A$,  is given as before,
 \label{prop3}
\begin{equation*}
A : \ket{x}   \ket{v}  \ket{m_0 \ldots m_{N-1}}  \mapsto \ket{x+v}   \ket{v} \ket{m_0 \ldots m_{N-1}}.
\end{equation*}
The overall transition is
\begin{equation*}
T =  A    S.
\end{equation*}

In this model, the parameter $\theta_v$  determines the effect of the velocity on the neighborhood memory qubits  through the memory operation $M$.  Also, there are $4$ possible neighborhood memory states,  hence the back-action on velocity through the ricochet $R$ is  determined by the  $4$ back-action parameters $\{\theta_{\mathbf{m}}\}$. These give  the walk a higher degree of maneuverability than the trivial neighborhood memory-ricochet model. The fact that the two stages act independently, however, implies that the particle (velocity) has no ``awareness" of the neighborhood memory state as it operates on them  through $M$. Similarly, the neighborhood memory state acts on the velocity through $R$  without ``awareness"  of the state of the velocity. Though this memory-ricochet scattering has more parameters than that with the trivial neighborhood, its general behavior   is expected to  remain  similar.%\footnote{We have not checked this claim through simulation, as we are primarily interested in the NHQW we discuss next. The interested reader may consider verifying it.}

%Let us compare with the trivial neighborhood cases we have alluded to.  First consider  $\theta_m=\pi/2$ and $\theta_b=0$ in the last section with trivial neighborhood. The memory step is the same if we let $\theta_v = \pi/2$ for both values of $p$. For the ricochet step, we equate the cases for the pair of  non-trivial neighborhood elements $s = m_{x-1}, m_{x+1} = 00,11$ (in~\eqref{memun2})  to the case  when $m_x = 0$ in the trivial neighborhood case~\eqref{memun}, otherwise $ s = m_{x-1}, m_{x+1} = 01,10$ is considered the same as when $m_x = 1$ previously. Thus, we set $\theta_{00}=\theta_{11}=\pi/4$ and $\theta_{01}=\theta_{10}=0$. It is clear with a little thought that the trajectories are  still straight lines for either direction, after some initial scattering.  
%
%To match the case in trivial neighborhood when $\theta_m=\pi/2$ and $\theta_b=\pi/2$, we choose $\theta_v = \pi/2$ for both values of $p$, $\theta_{00}=\theta_{11}=\pi/4$ and $\theta_{01}=\theta_{10}=\pi/2$. Again the situation is similar to the trivial neighborhood, but the particle moves two steps in the original direction between each loop back and forth. Hence the position spread  is larger. 

%By above, we see  that the behavior of this model will generally be the same as that of the trivial neighborhood memory-ricochet model.
We now create  an NHQW with  a scattering  built from  a basis of unentangled orthogonal vectors  of the Hilbert space in eq.~\eqref{eq:fqsq}, i.e., an {\it unentangled orthogonal basis} (UOB).\footnote{The characterization and construction of families of unentangled orthogonal bases (UOB) for multi-qubit systems is in~\cite{lsw:lduobag}.} We call it the {\it UOB-scattering} NHQW.

\subsection{UOB-scattering NHQW}  \label{subsec:gsmod}
We study a fundamentally different scheme for neighborhood dependence, one in which the scattering acts simultaneously on velocity and neighborhood memory.
%, with the same neighborhood $\mathcal{M}=\{-1,+1\}$. 

The neighborhood scattering operator, $U_S$, acts on  $\mathcal{H}_V \otimes  \mathcal{H}_{M_{x-1}} \otimes  \mathcal{H}_{M_{x+1}} $. %,   maps the velocity-neighborhood basis elements $\{\ket{v}\ket{m_{x-1} m_{x+1}}\}$ to  UOB. 
First,  we informally assign meaning to the neighborhood memory states $\{\ket{m_{x-1} m_{x+1}}\}$. We say that  $\ket{00}$, $\ket{11}$ are the  states from which  the particle scatters in a  ``balanced" manner, and  $\ket{01}$, $\ket{10}$ are the states from which the scattering is ``unbalanced". Note that the information about the scattering behavior is encoded in the pair $\ket{m_{x-1} m_{x+1}}$, and not individual memory qubits. Information about velocity relative to the neighborhood state is in the basis state $\ket{v}\ket{m_{x-1} m_{x+1}}$ of $\mathcal{H}_V \otimes  \mathcal{H}_{M_{x-1}} \otimes  \mathcal{H}_{M_{x+1}} $.
% . Let us try to define a reasonable  scattering scheme.

To ascribe explicitly the change experienced by the internal (spin)   velocity and neighborhood states from the scattering relative to  the pre-scattering basis states,   %It  encodes the information about  the state of the neighborhood memories relative to the direction of the particle, just as we have ascribed to the standard basis. 
  we construct a UOB that carries the part of the scattering information related to the internal (spin) states of  velocity and memory.
In the interest of  notational economy, we first encode the effect of scattering on the internal (spin) states through a parameterized  {\it spin-transformation} notation. A spin-transformation with parameter $\eta$ is  a map of a basis $\{\ket{b}, \ket{b}^\perp\}$ of $\mathbb{C}^2$ to %yields  transformation of a basis . An   of basis $\{\ket{b}, \ket{b}^\perp\}$ of $\mathbb{C}^2$, 
 the basis
\begin{alignat}{2} \label{rotatestuff}
\ket{b}_\eta& = \cos \eta \ket{b} + i \sin \eta \ket{b}^\perp,   \nonumber   \\
\ket{b}_\eta^\perp & = i\sin \eta \ket{b} + \cos \eta \ket{b}^\perp. 
\end{alignat}
We put this definition to use in defining the transformation of the internal (spin) states of velocity and memory. For velocity, we designate parameters $\alpha_0, \alpha_1$ for the spin states resulting from balanced scattering, and parameters $\beta_0, \beta_1$ for the unbalanced scatterings. Similarly  we define spin parameters $\gamma_l, \gamma_r$ (for left and right) that determine the behavior of the neighborhood memory qubits for both types of scatterings. The UOB is  defined using the spin parameters just described, and the spin-transformation in eq.~\eqref{rotatestuff},  
\begin{equation*}
\{\ket{j}_{\alpha_0}\ket{0}_{\gamma_l}\ket{1}_{\gamma_r}, \ket{j}_{\alpha_1}\ket{1}_{\gamma_l}\ket{0}_{\gamma_r}, \ket{j}_{\beta_0}\ket{1}_{\gamma_l}\ket{1}_{\gamma_r},  \ket{j}_{\beta_1}\ket{0}_{\gamma_l}\ket{0}_{\gamma_r} : j\in \{+1,-1\}\}.
%\{\ket{i}_{\alpha_0}\ket{j}_{\gamma_l}\ket{k}_{\gamma_r} : , j,k \in \{0,1\}\}
\end{equation*}
%To allow symmetric scattering  when the neighborhood states are $\ket{00}$ or $\ket{11}$,  
 We also introduce external scattering parameters $\theta_{0}, \theta_{1}$,  involved in interactions that transform pairs of UOB elements in a manner akin to the spin transformation.  Having defined the parameters and the UOB, we construct  the neighborhood scattering operator $U_S$,
\begin{alignat}{4} \label{qwscat}
U_{S} : 
&\ket{+1}\ket{00} &&\mapsto &  & \cos\theta_{0} \ket{+1}_{\alpha_0}\ket{0}_{\gamma_l}\ket{1}_{\gamma_r}  + i \sin\theta_{0} \ket{-1}_{\alpha_1}\ket{1}_{\gamma_l}\ket{0}_{\gamma_r},  \nonumber \\
&\ket{-1}\ket{00} &&\mapsto &  & i \sin\theta_{0} \ket{+1}_{\alpha_0}\ket{0}_{\gamma_l}\ket{1}_{\gamma_r}  + \cos\theta_{0} \ket{-1}_{\alpha_1}\ket{1}_{\gamma_l}\ket{0}_{\gamma_r}, \nonumber \\
&\ket{+1}\ket{01} &&\mapsto &  &\ket{-1}_{\beta_0}\ket{1}_{\gamma_l}\ket{1}_{\gamma_r}, \nonumber   \\  
&\ket{-1}\ket{10} &&\mapsto &  & \ket{+1}_{\beta_0}\ket{1}_{\gamma_l}\ket{1}_{\gamma_r}, \nonumber   \\
&\ket{+1}\ket{10} &&\mapsto & & \ket{+1}_{\beta_1}\ket{0}_{\gamma_l}\ket{0}_{\gamma_r}, \nonumber \\
&\ket{-1}\ket{01} &&\mapsto &  & \ket{-1}_{\beta_1}\ket{0}_{\gamma_l}\ket{0}_{\gamma_r}, \nonumber  \\  
&\ket{+1}\ket{11} &&\mapsto &  &\cos\theta_{1}  \ket{+1}_{\alpha_1}\ket{1}_{\gamma_l}\ket{0}_{\gamma_r} + i \sin\theta_{1} \ket{-1}_{\alpha_0}\ket{0}_{\gamma_l}\ket{1}_{\gamma_r}, \nonumber \\
&\ket{-1}\ket{11} &&\mapsto & &  i \sin\theta_{1} \ket{+1}_{\alpha_1}\ket{1}_{\gamma_l}\ket{0}_{\gamma_r}  +  \cos\theta_{1} \ket{-1}_{\alpha_0}\ket{0}_{\gamma_l}\ket{1}_{\gamma_r}.  \nonumber \\
\end{alignat}
%\begin{alignat}{4} \label{qwscat}
%U_S : 
%&\ket{+1}\ket{00} &&\mapsto &  & \cos\theta_{00} \ket{+1}_{\alpha_0}\ket{0}_l\ket{1}_r  + i \sin\theta_{00} \ket{-1}_{\alpha_1}\ket{1}_l\ket{0}_r,  \nonumber \\
%&\ket{-1}\ket{00} &&\mapsto &  & i \sin\theta_{00} \ket{+1}_{\alpha_0}\ket{0}_l\ket{1}_r  + \cos\theta_{00} \ket{-1}_{\alpha_1}\ket{1}_l\ket{0}_r, \nonumber \\
%&\ket{+1}\ket{01} &&\mapsto &  &\ket{-1}_{\beta_0}\ket{1}_l\ket{1}_r, \nonumber   \\  
%&\ket{-1}\ket{10} &&\mapsto &  & \ket{+1}_{\beta_0}\ket{1}_l\ket{1}_r, \nonumber   \\
%&\ket{+1}\ket{10} &&\mapsto & & \ket{+1}_{\beta_1}\ket{0}_l\ket{0}_r, \nonumber \\
%&\ket{-1}\ket{01} &&\mapsto &  & \ket{-1}_{\beta_1}\ket{0}_l\ket{0}_r, \nonumber  \\  
%&\ket{+1}\ket{11} &&\mapsto &  &\cos\theta_{11}  \ket{+1}_{\alpha_1}\ket{1}_l\ket{0}_r + i \sin\theta_{11} \ket{-1}_{\alpha_0}\ket{0}_l\ket{1}_r, \nonumber \\
%&\ket{-1}\ket{11} &&\mapsto & &  i \sin\theta_{11} \ket{+1}_{\alpha_1}\ket{1}_l\ket{0}_r  +  \cos\theta_{11} \ket{-1}_{\alpha_0}\ket{0}_l\ket{1}_r.  \nonumber \\
%\end{alignat}
%
Notice that the  scattering updates the neighboring memory qubits together with the velocity. This is a  different form of scattering compared with  the memory-ricochet model. It takes balanced neighborhoods to unbalanced neighborhoods and vice-versa. The post-scattering velocity is ``informed" by the pre-scattering neighborhood state and  velocity. Conceptually, this scattering has the goal of implementing a reasonable search strategy. 

In Figures~\ref{fig:ndqw5}-\ref{fig:ndqw8}, we show  simulations of the NHQW  for the parameter values  in Table~\ref{table1}.  The lattice size is $N=13$.  The initial state for each walk is ($\lfloor N/2 \rfloor$ is the center of the lattice) 
 \begin{align*}
 \psi_0 &= \frac{1}{\sqrt{2}} \ket{\lfloor N/2 \rfloor}(\ket{+1} + \ket{-1})\bigotimes^N \ket{0}\\
 &= \frac{1}{\sqrt{2}} \ket{6}(\ket{+1} + \ket{-1})\bigotimes^N \ket{0}.
 \end{align*}
Each walk simulation is run for $\lfloor N/2 \rfloor = 6$ time steps. We plot the probability distribution that  the particle is at position $x\in \mathbb{Z}_N$ at each time step (the  velocity and memory tensor factors are traced out).
\begin{table}   [H]
\begin{center}
\begin{tabu}to\linewidth{|[1.5pt]c|c|c|c|c|c|c|c|c|[1.5pt]}  %
\tabucline[1.5pt]-
& $\alpha_0$ & $\alpha_1$ & $\beta_0$ & $\beta_1$ & $\gamma_l$ & $\gamma_r$ & $\theta_0$ & $\theta_1$ \\
\tabucline[1.5pt]-
Fig.~\ref{fig:ndqw5}&$0$& $0$& $\pi/4$& $\pi/4$& $\pi/2$& $\pi/2$& $\pi/4$ &$\pi/4$  \\
\hline 
Fig.~\ref{fig:ndqw4}&$0$& $0$& $0$& $0$& $0$& $0$& $\pi/4$ &$\pi/4$  \\
\hline 
Fig.~\ref{fig:ndqw7}&$0$& $0$& $0$& $0$& $0$& $0$& $0$ &$0$ \\ 
\hline 
Fig.~\ref{fig:ndqw6}&$\pi/2$& $0$& $0$& $0$& $0$& $0$& $0$ &$0$  \\ 
\hline 
Fig.~\ref{fig:ndqw2}&$\pi$& $\pi/3$& $\pi$& $\pi/6$& $\pi/2$& $\pi/2$& $\pi/6$ &$\pi/2$ \\ 
\hline 
Fig.~\ref{fig:ndqw1}&$0$& $\pi/2$& $0$& $\pi/2$& $0$& $0$& $0$ &$\pi/4$ \\ 
\hline 
Fig.~\ref{fig:ndqw3}&$\pi/2$& $0$& $0$& $\pi/2$& $0$& $\pi/2$& $0$ &$\pi/4$ \\ 
\hline 
Fig.~\ref{fig:ndqw8}&$0$& $\pi/2$& $0$& $\pi/2$& $0$& $\pi/6$& $\pi/4$ &$\pi/4$ \\ 
\tabucline[1.5pt]-
\end{tabu}
\end{center}
\caption{Parameter values used in simulations of NHQW}
 \label{table1}  
\end{table} 

Figure~\ref{fig:ndqw5} reproduces the classical random walk, while Figure~\ref{fig:ndqw4} shows the usual quantum walk. 
%
%\begin{figure}[H] 
%\includegraphics{figqw1.pdf}
%\caption{A QLGA model of QW with site-history.}
%\label{figure1}
%\end{figure}
%
\begin{figure}[H]
\minipage{0.5\textwidth} % \label{fig:ndqw5}
  \includegraphics[width=0.85\linewidth]{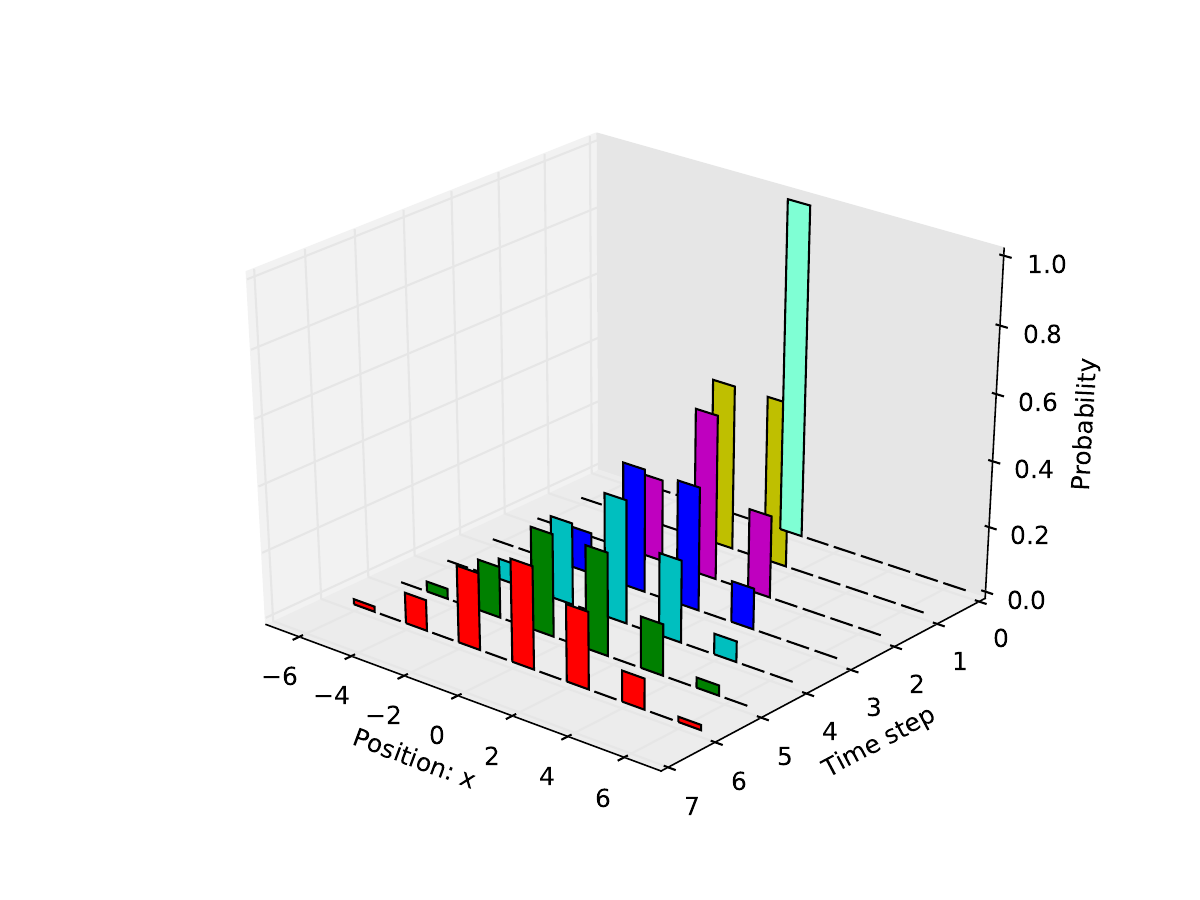} 
     \caption{
     Classical random walk:
$\alpha_0=
\alpha_1=0,
\beta_0=
\beta_1=
\theta_0=
\theta_1=\pi/4,
\gamma_l=
\gamma_r=\pi/2$.
} 
 \label{fig:ndqw5}
\endminipage\hfill
\minipage{0.5\textwidth}% \label{fig:ndqw4}
  \includegraphics[width=0.85\linewidth]{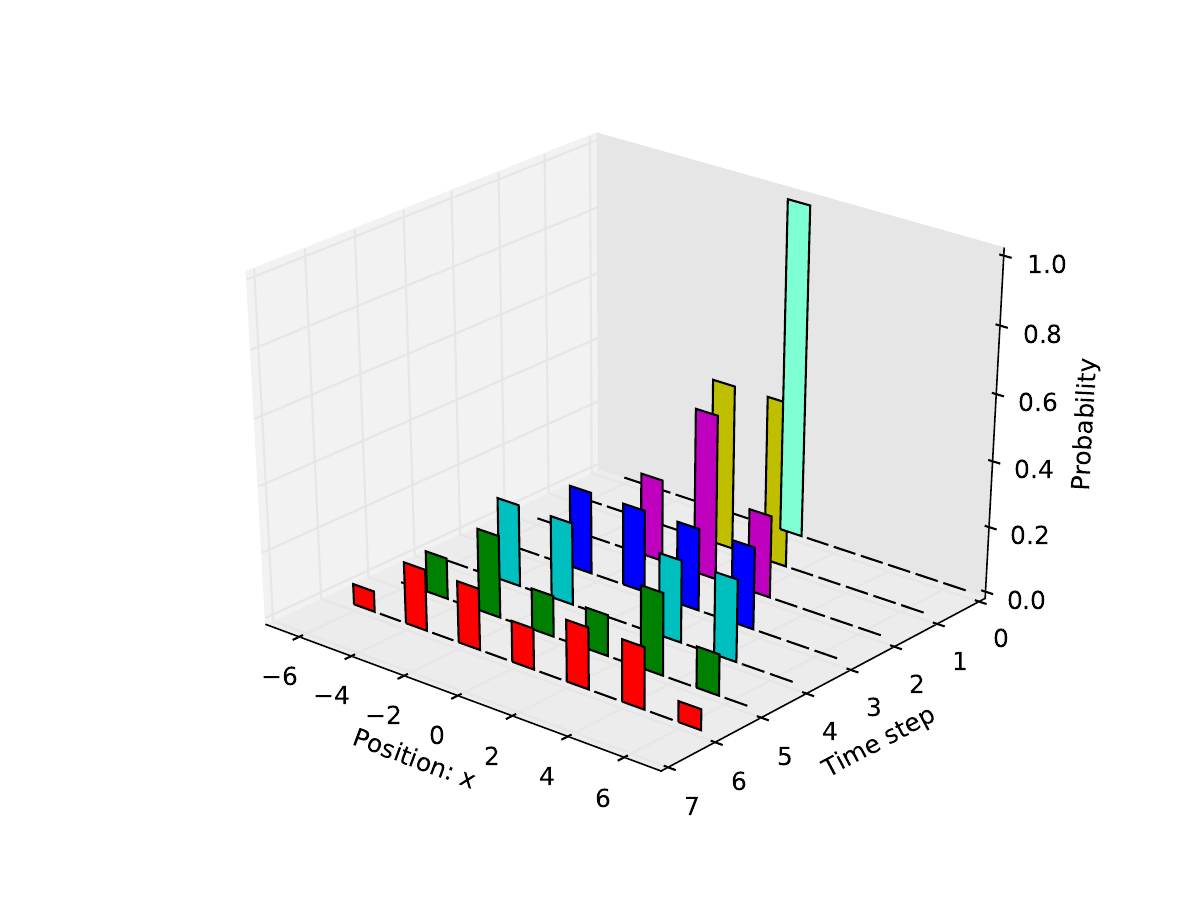} 
     \caption{Standard quantum  walk:
     $\alpha_0=
\alpha_1=
\beta_0=
\beta_1=
\gamma_l=
\gamma_r=0,
\theta_0=
\theta_1=\pi/4$.
     }
     \label{fig:ndqw4} 
 \endminipage
%   \caption{QW evolution} \label{fig:ndqw12}
\end{figure}
The next few walks depart significantly from the usual classical random or quantum walks. Figure~\ref{fig:ndqw7} shows a walk in which the particle continues on its straight line  trajectory it was initially set to while changing the memory qubits as it walks. Figure~\ref{fig:ndqw6} differs from this by flipping the $\ket{+1}$ velocity to $\ket{-1}$ at the start and then continuing with the straight line trajectory.
\begin{figure}[H]
\minipage{0.5\textwidth}%
  \includegraphics[width=0.85\linewidth]{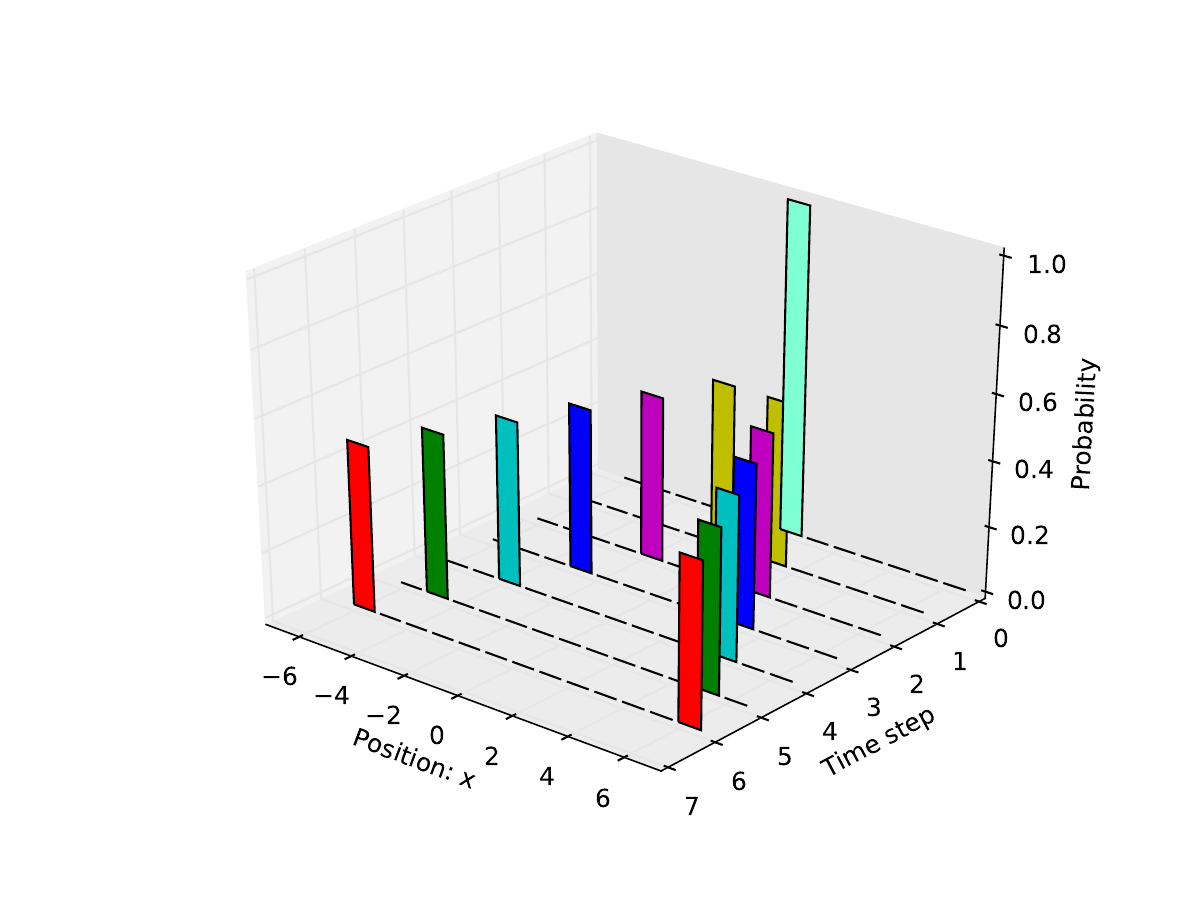} 
   \caption{Straight line walk while flipping memory qubits:
     $\alpha_0=\alpha_1=
\beta_0=
\gamma_l=
\gamma_r=
\theta_0=
\theta_1=0$.
} 
%     \caption{Straight line walk while flipping memory qubits}
      \label{fig:ndqw7}
 \endminipage
\minipage{0.5\textwidth} % 
  \includegraphics[width=0.85\linewidth]{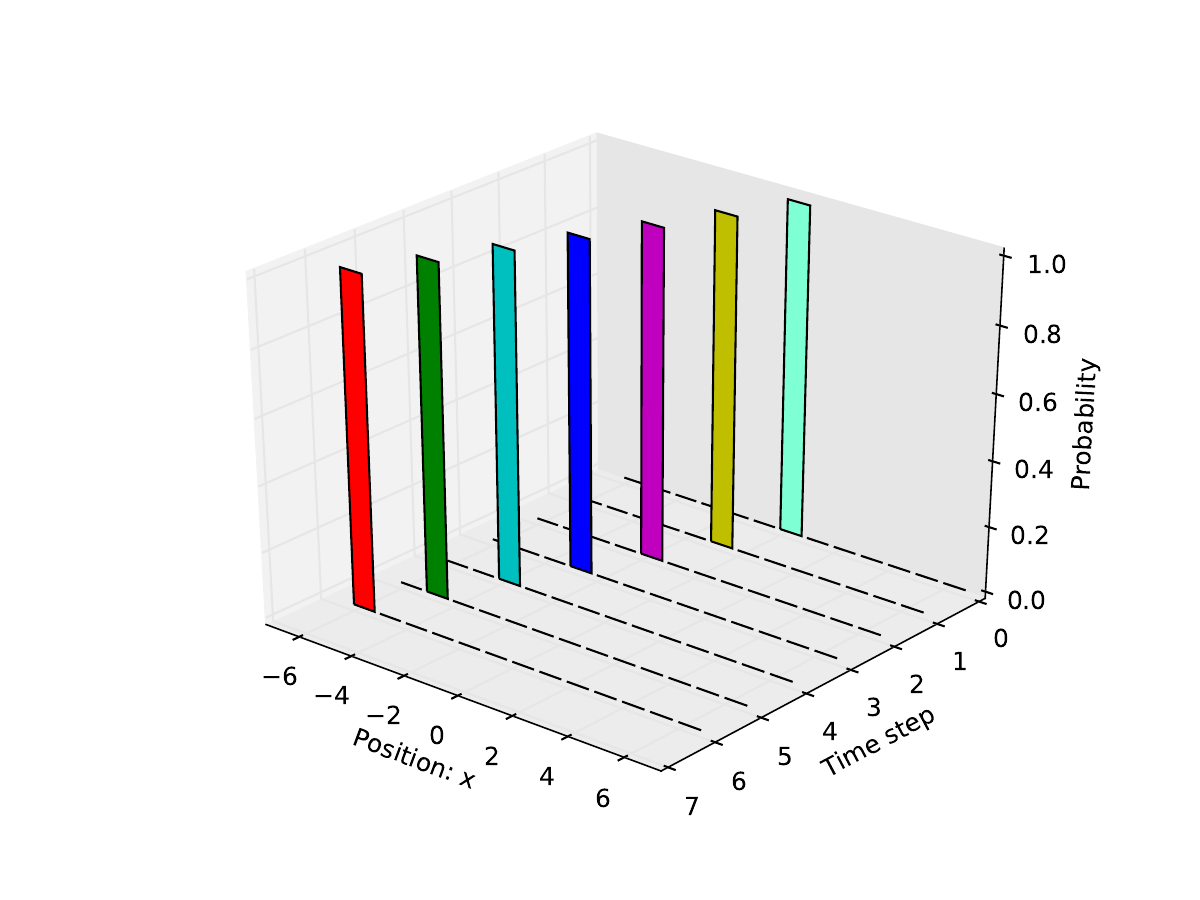} 
     \caption{Initial flip of velocity then straight line walk:
     $\alpha_0=
\beta_0=
\gamma_l=
\gamma_r=
\theta_0=
\theta_1=0, \alpha_1=\pi/2$.
}
%     \caption{Initial flip of velocity then straight line walk}
\label{fig:ndqw6}
 \endminipage
 \end{figure}
Figures~\ref{fig:ndqw2},~\ref{fig:ndqw1},~\ref{fig:ndqw3},~\ref{fig:ndqw8}, show increasingly complex patterns that result from interaction of parameters. These walks demonstrate  behaviors displaying  flexibility, directionality and quantum randomness.
\begin{figure}[H]
\minipage{0.5\textwidth} %
  \includegraphics[width=0.85\linewidth]{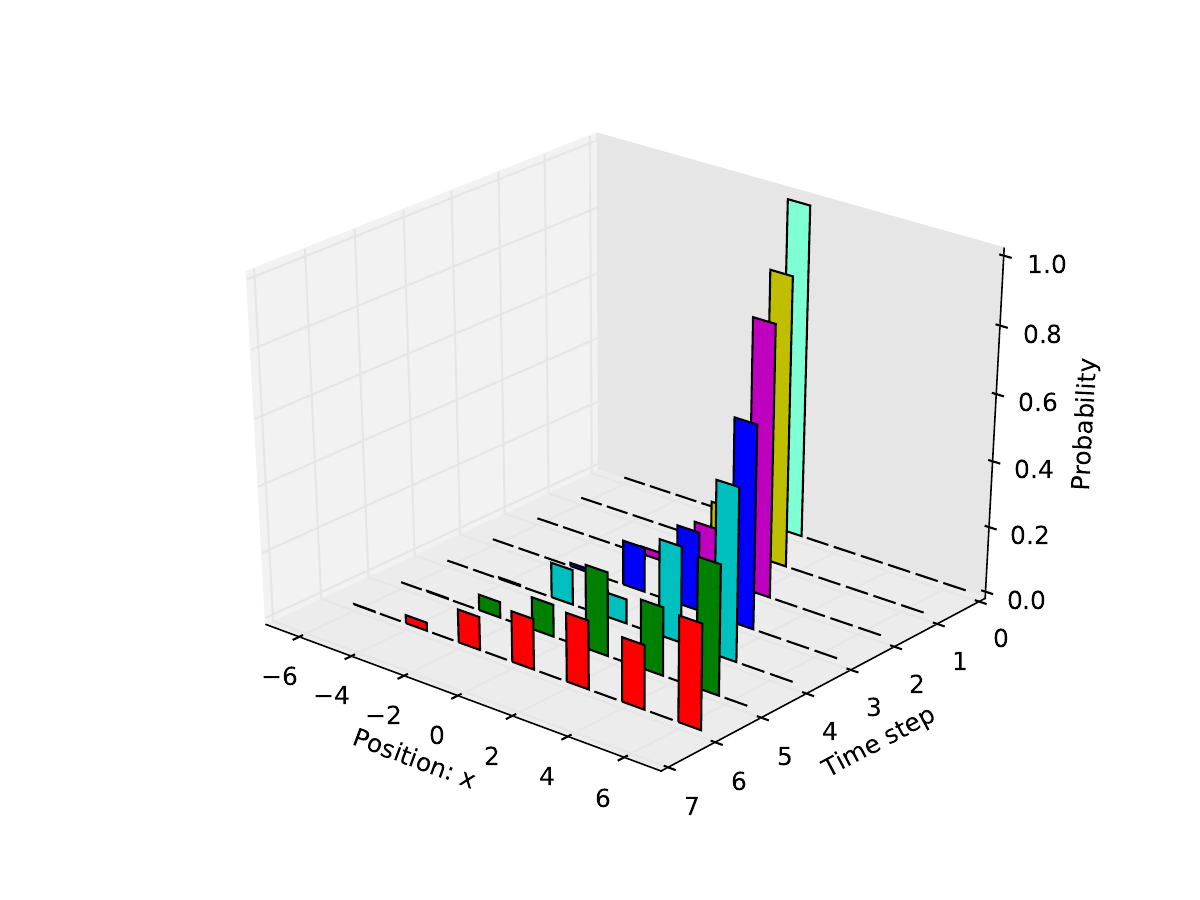} 
     \caption{Directionality and spread in walk:
  $\alpha_0=\pi,
\alpha_1=\pi/3,
\beta_0=\pi,
\beta_1=\pi/4,
\gamma_l=\pi/6,
\gamma_r=\pi/2,
\theta_0=\pi/6,
\theta_1=\pi/2$.
} 
%     \caption{Directionality and spread in walk.}
\label{fig:ndqw2}
\endminipage\hfill
\minipage{0.5\textwidth} %
  \includegraphics[width=0.85\linewidth]{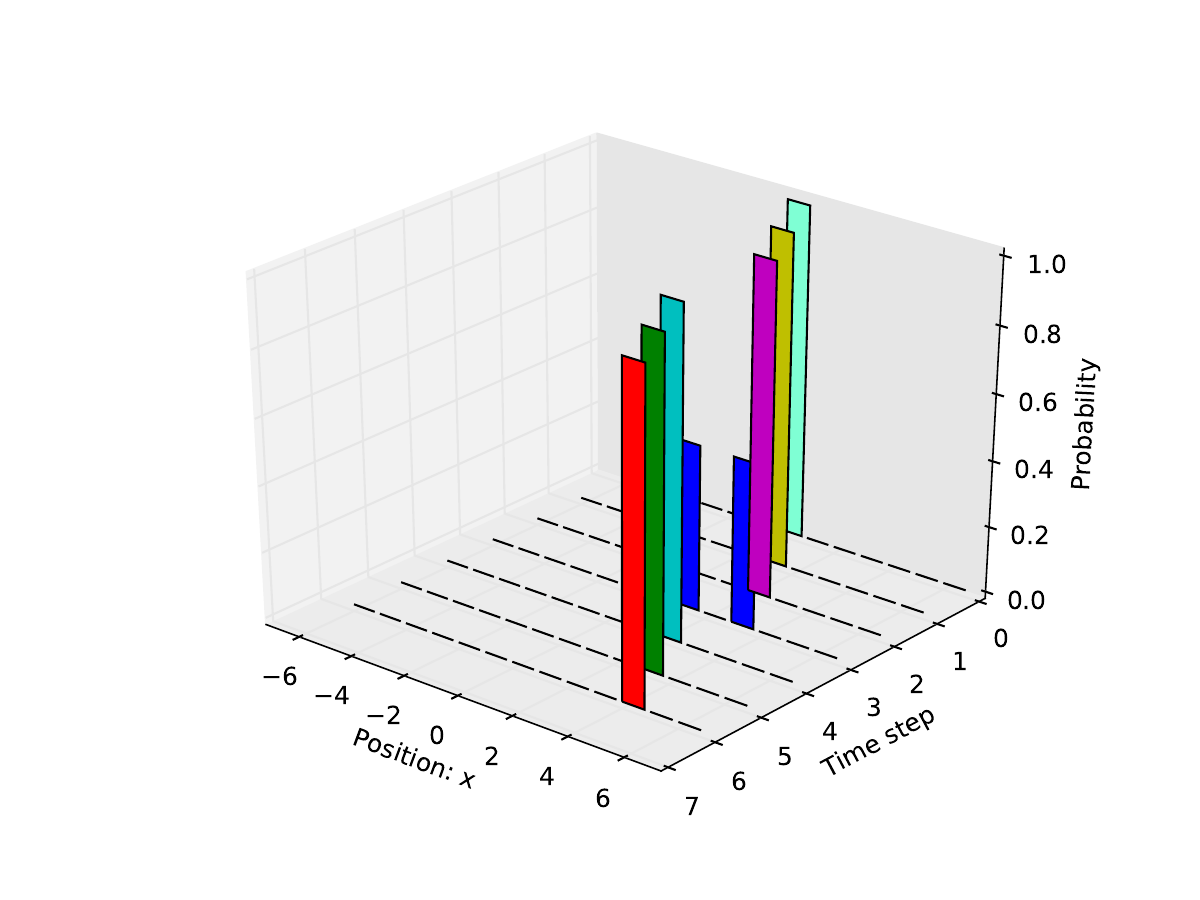} 
  \caption{Directionality and spread in walk:
 $\alpha_0=
\beta_0=
\gamma_l=
\gamma_r=
\theta_0=0,
\alpha_1=\pi/2=
\beta_1=\pi/2,
\theta_1=\pi/4$.
}
%     \caption{Directionality and spread in walk.}
\label{fig:ndqw1} 
\endminipage 
\end{figure}
\begin{figure}[H]
\minipage{0.5\textwidth} %
  \includegraphics[width=0.85\linewidth]{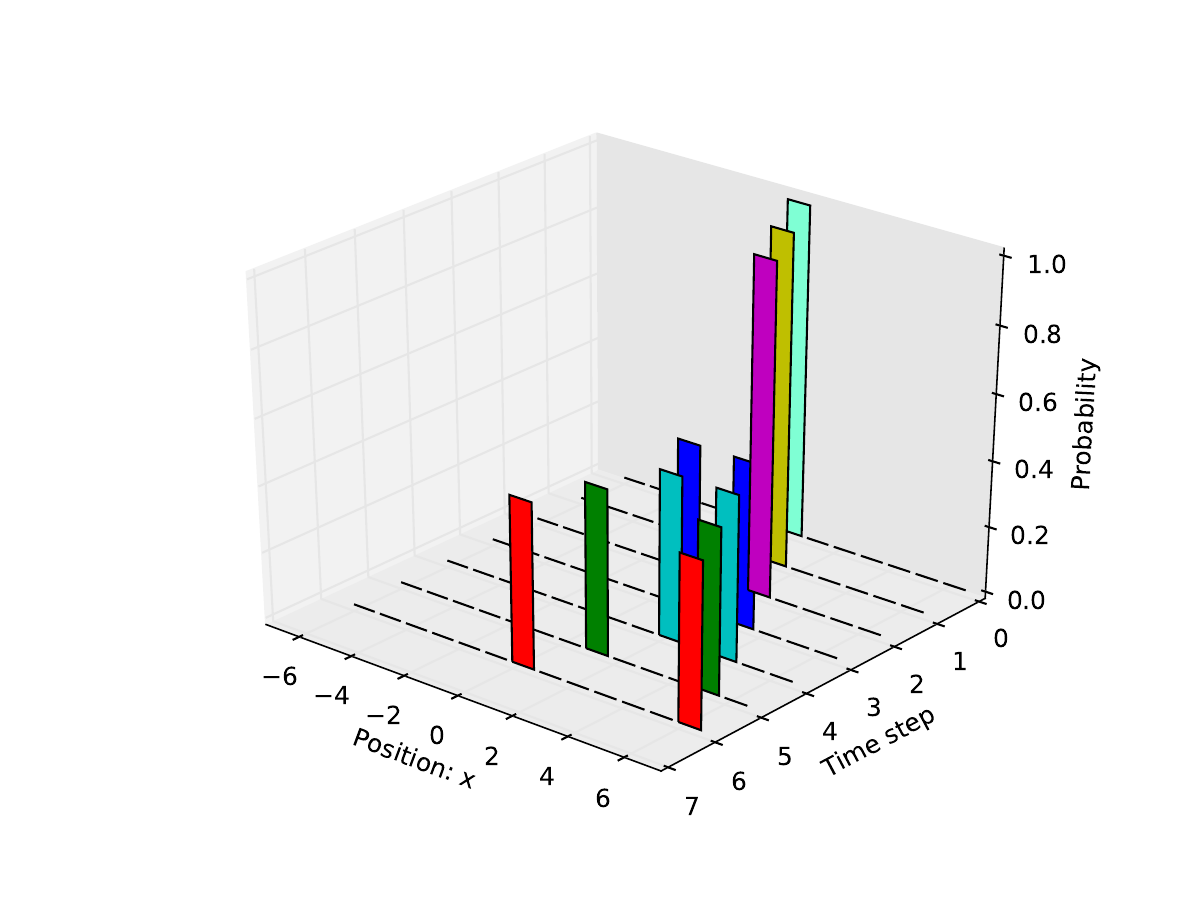} 
   \caption{Directionality and spread in walk:
$\alpha_1=
\beta0=
\gamma_l=
\theta_0=0,
\alpha_0=
\beta_1=
\gamma_r=\pi/2,
\theta_1=\pi/4$.
}
%     \caption{Directionality and spread in walk $3$.}
\label{fig:ndqw3}
 \endminipage
 \hfill
\minipage{0.5\textwidth} %
\includegraphics[width=0.85\linewidth]{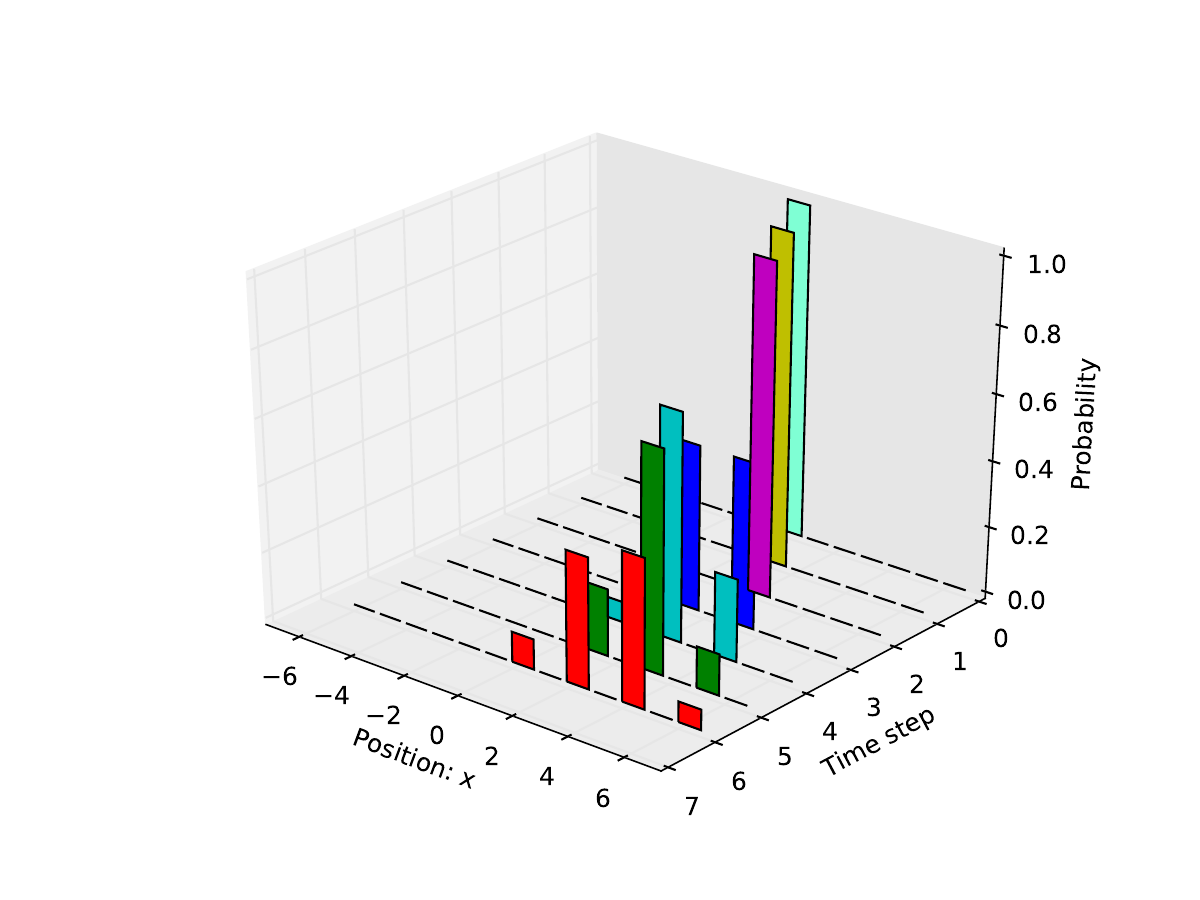} 
\caption{Directionality and spread in walk:
$\alpha_0=
\beta_0=
\gamma_l=0,
\alpha_1=
\beta_1=\pi/2,
\gamma_r=\pi/6,
\theta_0=
\theta_1=\pi/4$.
}
%     \caption{Directionality and spread in walk $4$.}
\label{fig:ndqw8}
 \endminipage
\end{figure}

\subsection{UOB-scattering NHQW as a single-particle sector of a QCA}  \label{subsubsec:gsmod}
A quantum cellular automaton (QCA) is a discrete space, discrete time, model of physics. It consists of  a collection of cells (or sites)  on a lattice, each with an identical finite-dimensional Hilbert space over it, called the {\it cell Hilbert space}.  The state of the QCA is a density operator on  the Hilbert space of the QCA, which  is the tensor product of the cell Hilbert spaces.  The evolution of the state of the QCA  is described by a {\it global  evolution operator}  that is unitary, translation-invariant and causal, i.e., restricts information to travel at a finite speed.  This means that the evolution is restricted to  propagate information from  a cell to  others  within a  finite {\it neighborhood} of it at each step of QCA evolution. 

A quantum lattice gas automaton (QLGA) is  a special kind of QCA, with  multiple quantum particles in each cell. If we  a particle may  either be present or absent in a cell, it is  represented by a qubit ($\mathbb{C}^2$), with the state $\ket{0}$ being ``particle absent", and $\ket{1}$ being ``particle present". If multiple particles are involved,  the Hilbert space of a cell is  a tensor product of multiple $\mathbb{C}^2$ factors (qubits), each of which corresponds to a  particle. %The state of each  particle in a cell is represented by an element of the respective tensor factor. 

The dynamics (evolution) of a QLGA consist of  advection followed by scattering (or in the opposite order). These are  similar, in spirit,  to a QW but in  a multi-particle setting.  Advection, the propagation (hopping) of particles  in  a QLGA is a permutation of tensor factors among neighboring cell Hilbert spaces. Advection operator takes the tensor factors of a cell Hilbert space to the corresponding tensor factors of one of its neighbors, thus carrying the information about the state of particles from a cell to another.  Scattering operator in a QLGA acts   as an identical  unitary operator on each cell Hilbert space. This cell-wise  operation is said to be  {\it local}. The reader is referred to~\cite{bib:slwqcaqlga, bib:sdlhdQWqlga, bib:msqcawp} for the definition and examples of  QCA and QLGA. 

A single-particle  QLGA state has  one cell with a single particle in it, and the rest without a particle. The  single-particle sector is the  span of the  single-particle states. In~\cite{bib:sdlhdQWqlga}, a number of  history dependent QWs were shown by construction to be  the single-particle sectors of  QLGA.    The NHQW described in this paper, however,  needs a more general description by a QCA for its multi-particle version. For a QCA, there is no natural definition of a particle.  However, a cell Hilbert space may still be composed of multiple qubits, i.e., multiple $\mathbb{C}^2$ tensor factors. We let some of these  represent particles and others represent memory qubits, a distinction that serves our purpose. A single-particle state has exactly one cell with a single particle in it, and none in the others.  The span of a {\it subset} of  single-particle states is a single-particle sector of a QCA~\cite{bib:slwqcaqlga}.
%s 

\begin{thm}
The UOB-scattering NHQW is a single-particle sector of a QCA.
\end{thm}
\begin{proof}
We construct the QCA and its specific single-particle sector  equivalent to the UOB-scattering NHQW.  A cell of this QCA consists of $4$ qubits,  $V_0 \otimes V_1 \otimes M_0 \otimes M_1 =  \bigotimes^4 \mathbb{C}^2$. Two of these  $V_0, V_1 = \mathbb{C}^2$, hold the particle state as an element of $V_0 \otimes V_1$. A cell may hold none:  $\ket{v_0 v_1} \in \ket{00}$, one (single):  $\ket{v_0 v_1} \in \{\ket{01}, \ket{10}\}$, or two: $\ket{v_0 v_1} \in \ket{11}$, particles. The other two qubits  $M_0,  M_1 = \mathbb{C}^2$, hold the memory state of the cell  as an element of $M_0 \otimes M_1$. A basis state of a cell  is $\ket{v_0 v_1}  \ket{m_0 m_1} : v_i,m_i  \in \{0,1\}$. Basis states of the QCA are  $\bigotimes_{k\in \mathbb{Z}_N} \ket{v^k_0 v^k_1} \ket{m^k_0 m^k_1}$, where the superscript $k$ indicates the cell index.
\begin{figure}[H] 
\includegraphics[width=0.75\textwidth]{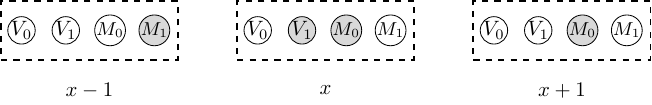} %[width=0.95\textwidth]
\caption{Cells of the multi-particle QCA generalizing the UOB-scattering NHQW.  White circles represent $\ket{0}$. Gray circles represent $\ket{1}$.}
\label{figure2}
\end{figure}
Let us  describe how the NHQW states are to be interpreted in the QCA. Each single-particle state of the cell is assigned a direction of motion: $\ket{v_0 v_1} = \ket{01}$ is right-moving and $\ket{v_0 v_1} = \ket{10}$ is left-moving. 
%single-particle right/left-moving  state of the QCA is given by embedding 
%\begin{alignat*}{3}
%\ket{x}\ket{+1}  &\longleftrightarrow  &\underbrace{\ket{01}}_x\otimes \bigotimes_{k\neq x} \ket{00} \in \bigotimes_{k\in \mathbb{Z}_N} V_0\otimes V_1\\
%\ket{x}\ket{-1}  &\longleftrightarrow  &\underbrace{\ket{10}}_x\otimes \bigotimes_{k\neq x} \ket{00}  \in \bigotimes_{k\in \mathbb{Z}_N} V_0\otimes V_1\
%\end{alignat*}
%
%If the particle is at $x$, single-particle memory state associated to it is obtained by embedding
%\begin{alignat*}{3}\bigotimes_{k\in \mathbb{Z}_N} \ket{m}_k  &\longleftrightarrow  &\bigotimes_{k < x}\ket{0 m_0}_k\otimes \bigotimes_{k \geq x}\ket{m_0 0}_k   \in \bigotimes_{k\in \mathbb{Z}_N} M_0\otimes M_1
%\end{alignat*}
A single-particle sector of the QCA is the subspace spanned by basis states  having only one cell $x$ in a  single-particle  state  $\ket{v^x_0 v^x_1} \in \{\ket{01}, \ket{10}\}$, while the rest of the cells are in the no-particle state $\ket{v^k_0 v^k_1} = \ket{00}$ for $k\ne x$. Memory qubits of the NHQW are inserted in the QCA states so as to allow the dynamics of the QCA to affect the NHQW transitions correctly. An NHQW basis  state is embedded in the QCA Hilbert space as
\begin{alignat}{3} \label{stateemb}
\ket{x}\ket{+1}\otimes \bigotimes_{k\in \mathbb{Z}_N} \ket{m_k}  &\longleftrightarrow  &\bigotimes_{k < x}\ket{00}\ket{0 m_k}\otimes \underbrace{\ket{01}\ket{m_x 0}}_x \otimes  \bigotimes_{k > x}\ket{00}\ket{m_k 0}, \nonumber \\
\ket{x}\ket{-1}\otimes \bigotimes_{k\in \mathbb{Z}_N} \ket{m_k}   &\longleftrightarrow  &\bigotimes_{k < x}\ket{00}\ket{0 m_k}\otimes \underbrace{\ket{10}\ket{m_x 0}}_x \otimes  \bigotimes_{k > x}\ket{00}\ket{m_k 0} .
\end{alignat}

 %The particle states $\ket{01}, \ket{10}$ represent the left and right moving NHQW particle states. In other words, we encode the particle position and velocity $\ket{x}\ket{+1}, \ket{x}\ket{-1}$ together in $\ket{v_0 v_1}_x = \ket{01}, \ket{10}$, respectively.   

%The NHQW memory state  $\bigotimes{x\in \mathbb{Z}_N} \ket{m_x}$ is contained in the elements $\bigotimes_{k < x}\ket{m_0}_x\otimes \bigotimes_{k \geq x}\ket{m_0}_x$. Rest of the memory elements of the QCA are set as $\bigotimes_{k < x}\ket{m_0}_x\otimes \bigotimes_{k \geq x}\ket{m_0}_x = \bigotimes_{k < x}\ket{0}\otimes \bigotimes_{k \geq x}\ket{0}$. This  strict subspace of the single-particle sector of the QCA is the NHQW, our construction providing a bijection between the basis states of NHQW and the vectors in this subspace. %The state of the NHQW  can, therefore, be obtained within the single-particle sector, for a basis state with particle at $x$,  by tracing out the elements $\bigotimes_{k < x}\ket{m_0}_x\otimes \bigotimes_{k \geq x}\ket{m_0}_x$.
%Consider  a single-particle basis state of the QCA with site $x$ containing the particle.
%The dynamics  of the QCA that implements the NHQW are shown in figure~\ref{figure2}. 
 
%For each QLGA stage, part of the advection with center cell $x$  as the origin  is shown.
%\begin{figure}[H] 
%\includegraphics{figqw0.pdf}
%\caption{QCA model of  NHQW.}
%\label{figure2}
%\end{figure}

Evolution of this QCA  consists of three QLGA in tandem.  We denote the first QLGA's (stage $1$) advection as  $\sigma_1$ and   scattering as $S_1$. It uses $\sigma_1$ to shuffle the neighboring memory elements to the center cell. 
\begin{equation*}
\sigma_1: \bigotimes_{k\in \mathbb{Z}_N} \ket{m^k_0}\ket{m^k_1} \mapsto \bigotimes_{k\in \mathbb{Z}_N} \ket{m^{k+1}_0}\ket{m^{k-1}_1},
\end{equation*} 
acting as identity on the $V_1, V_2$ factors of each cell. Then $S_1$ acts cell-by-cell,  restricting  to each cell as a  local (cell-wise) scattering  $L_1$ that mimics the NHQW neighborhood scattering operator  $U_S$ in eq.~\eqref{qwscat}. The scattering,  $U_{S}$, of  the NHQW carries over to the local scattering, $L_1$, of the QLGA,   through replacing the velocity states of NHQW, $\ket{+1}$, $\ket{-1}$,  with  the right/left moving states $\ket{01}, \ket{10} \in V_0 \otimes V_1$ of the QCA, respectively. The spin-transformation with parameter $\eta$  in eq.~\eqref{rotatestuff} acts on $\{\ket{01}, \ket{10}\}$ as 
\begin{alignat*}{2} % \label{rotatemorestuff}
\ket{01}_\eta & = \cos \eta \ket{01} + i \sin \eta \ket{10},   \nonumber   \\
\ket{10}_\eta & = i\sin \eta \ket{01} + \cos \eta \ket{10}. 
\end{alignat*}
We now get a straightforward description of the first local scattering  $L_1$ that   acts as  the  NHQW neighborhood scattering operator  $U_S$,
\begin{alignat*}{4} %\label{qwscat0}
L_1 : & V_0\otimes V_1 \otimes M_0 \otimes M_1 &\quad &\mapsto & \quad & V_0\otimes V_1 \otimes M_0 \otimes M_1 \nonumber \\
&\ket{01}\ket{00} &&\mapsto &  & \cos\theta_{0} \ket{01}_{\alpha_0}\ket{0}_{\gamma_l}\ket{1}_{\gamma_r}  + i \sin\theta_{0} \ket{10}_{\alpha_1}\ket{1}_{\gamma_l}\ket{0}_{\gamma_r},  \nonumber \\
&\ket{10}\ket{00} &&\mapsto &  & i \sin\theta_{0} \ket{01}_{\alpha_0}\ket{0}_{\gamma_l}\ket{1}_{\gamma_r}  + \cos\theta_{0} \ket{10}_{\alpha_1}\ket{1}_{\gamma_l}\ket{0}_{\gamma_r}, \nonumber \\
&\ket{01}\ket{01} &&\mapsto &  &\ket{10}_{\beta_0}\ket{1}_{\gamma_l}\ket{1}_{\gamma_r}, \nonumber   \\  
&\ket{10}\ket{10} &&\mapsto &  & \ket{01}_{\beta_0}\ket{1}_{\gamma_l}\ket{1}_{\gamma_r}, \nonumber   \\
&\ket{01}\ket{10} &&\mapsto & & \ket{01}_{\beta_1}\ket{0}_{\gamma_l}\ket{0}_{\gamma_r}, \nonumber \\
&\ket{10}\ket{01} &&\mapsto &  & \ket{10}_{\beta_1}\ket{0}_{\gamma_l}\ket{0}_{\gamma_r}, \nonumber  \\  
&\ket{01}\ket{11} &&\mapsto &  &\cos\theta_{1}  \ket{01}_{\alpha_1}\ket{1}_{\gamma_l}\ket{0}_{\gamma_r} + i \sin\theta_{1} \ket{10}_{\alpha_0}\ket{0}_{\gamma_l}\ket{1}_{\gamma_r}, \nonumber \\
&\ket{10}\ket{11} &&\mapsto & &  i \sin\theta_{1} \ket{01}_{\alpha_1}\ket{1}_{\gamma_l}\ket{0}_{\gamma_r}  +  \cos\theta_{1} \ket{10}_{\alpha_0}\ket{0}_{\gamma_l}\ket{1}_{\gamma_r}.  \nonumber \\
\end{alignat*}
On all the other basis elements $L_1$ acts as the identity. The   scattering operator $S_1$ of  stage $1$ QLGA acts by $L_1$ on each cell, so it is:
\begin{equation*}
S_1 = \bigotimes_{k\in \mathbb{Z}_N} L_1
\end{equation*}
The  global evolution operator for this QLGA is
\begin{equation*}
\mathcal{G}_1 = S_1 \sigma_1.
\end{equation*}

The next two QLGA stages are needed to accomplish the NHQW advection\footnote{We use the term {\it advection} both for a walking step of the NHQW and  for the propagation of multiple  particles in a  QLGA.} and switch the memory qubits to valid positions. Stage $2$ QLGA is described by the advection $\sigma_2=\sigma_1^{-1}$ and local scattering $L_2$. $\sigma_2$ takes the  memory qubits from cell $x$ (these were shuffled in from the neighbors by $\sigma_1$, and then altered by  $L_1$) back to the respective neighbors. In the process it also returns the memory qubits of cell  $x$ from the neighbors (unaltered). 
\begin{alignat*}{4}
\sigma_2 = \sigma^{-1}_1: &\bigotimes_{k\in \mathbb{Z}_N} \ket{m^k_0}\ket{m^k_1} &&\mapsto &  & \bigotimes_{k\in \mathbb{Z}_N} \ket{m^{k-1}_0}\ket{m^{k+1}_1}. \\
\end{alignat*}
Stage $2$ QLGA's local scattering $L_2$, which sets the center cell $x$ memory to the correct position before the next QLGA  (stage $3$) advection sends it  to its destination (cell $x+1$). This is needed to ensure that the memory will be in the correct form, prescribed by eq.~\eqref{stateemb}, at the end of  the current  QCA evolution step (after stage $3$ QLGA).  $L_2$  conditionally switches the states of $M_0$ and $M_1$ if there is a right-moving particle in the cell.
\begin{alignat*}{4}
%\sigma^{-1}_1: &\bigotimes_{x\in \mathbb{Z}_N} \ket{m_0}_x\ket{m_1}_x &&\mapsto &  & \bigotimes_{x\in \mathbb{Z}_N} \ket{m_0}_{x-1}\ket{m_1}_{x+1}, \\
%& & & & & \\
%S_2 : & V_0\otimes V_1 \otimes M_0 \otimes M_1 &&\mapsto &  & V_0\otimes V_1 \otimes M_0 \otimes M_1 \nonumber \\
L_2 : &\ket{01}\ket{m_0 m_1} &&\mapsto &  &\ket{01}\ket{m_1 m_0}. \nonumber  % \\  
%&\ket{10}\ket{m_0 m_1} &&\mapsto &  & \ket{10}\ket{m_0 m_1}  \nonumber
\end{alignat*}
It acts  as identity on the other basis elements. The stage $2$ scattering operator, $S_2$, is
\begin{equation*}
S_2 = \bigotimes_{k\in \mathbb{Z}_N} L_2
\end{equation*}
The  evolution operator for this QLGA is
\begin{equation*}
\mathcal{G}_2 = S_2 \sigma_2 = S_2 \sigma^{-1}_1 
\end{equation*}

The final stage QLGA (stage $3$) first uses an advection $\sigma_3$ to carry out the NHQW advection $A$,  hopping the particle right/left, 
\begin{equation*}
\sigma_3: \bigotimes_{k\in \mathbb{Z}_N} \ket{v^k_0}\ket{v^k_1} \mapsto \bigotimes_{k\in \mathbb{Z}_N} \ket{v^{k+1}_0}\ket{v^{k-1}_1},
\end{equation*}
while acting as identity on $M_0, M_1$ factors.
Then it applies a local scattering $L_3$, which sets the memory at the left destination cell (cell $x-1$) to its valid state to prepare for the next QCA evolution step. It  conditionally switches the states of  $M_0$ and $M_1$ if there is a left-moving particle in the cell, 
\begin{alignat*}{4}
%\sigma_3: &\bigotimes_{k\in \mathbb{Z}_N} \ket{v_0}_x\ket{v_1}_x &&\mapsto & &\bigotimes_{x\in \mathbb{Z}_N} \ket{v_0}_{x+1}\ket{v_1}_{x-1} \\
%&&&&& \\
%S_3 : & V_0\otimes V_1 \otimes M_0 \otimes M_1 &&\mapsto &  & V_0\otimes V_1 \otimes M_0 \otimes M_1 \nonumber \\
L_3 : %&\ket{01}\ket{m_0 m_1} &&\mapsto &  &\ket{01}\ket{m_0 m_1} \nonumber   \\  
&\ket{10}\ket{m_0 m_1} &&\mapsto &  & \ket{10}\ket{m_1 m_0},  \nonumber
\end{alignat*}
acting as identity on the other basis elements. This  ensures that the memory will be in the correct form prescribed by eq.~\eqref{stateemb}.  The stage $3$ scattering operator, $S_3$, is
\begin{equation*}
S_3 = \bigotimes_{k\in \mathbb{Z}_N} L_3
\end{equation*}The  evolution operator for stage $3$ QLGA is
\begin{equation*}
\mathcal{G}_3 = S_3 \sigma_3 
\end{equation*}

The global QCA evolution operator, whose restriction to the single-particle sector is the NHQW transition, is 
\begin{equation} \label{qcaevolv}
\mathcal{G} = \mathcal{G}_3 \mathcal{G}_2 \mathcal{G}_1 = S_3 \sigma_3 S_2 \sigma_2 S_1 \sigma_1 = S_3 \sigma_3 S_2 \sigma^{-1}_1 S_1 \sigma_1.
\end{equation}
\end{proof}
In the following figures we show one step of the NHQW transition, starting with a right-moving particle state and an unbalanced neighborhood memory configuration.  Figure 9 shows the QCA cells and an initial configuration of memory and particle states. After the NHQW scattering  counterpart $S_1$, we track  the  projection on either the right (Figure~\ref{figqw4}) or the left-moving  (Figure~\ref{figqw2}) single-particle sector, to  observe the dynamics.  The figures only show parts of the advection of each QLGA stage relevant to the center cell $x$ and its neighbors, i.e., the hops in which  cell $x$  plays a part as either  an origin or a destination or both during the evolution. White circles represent $\ket{0}$, gray circles represent $\ket{1}$, whereas colored circles represent any state that may be consistent with the QCA description.  

Figure~\ref{figqw4} shows the projection on the right-moving single-particle sector. After stage $2$ QLGA advection, it  shows the memory repositioning by stage $2$ scattering $S_2$ in center cell $x$ (origin).  At that point, the local scattering operator $L_2$  (local part of $S_2$), acting on the center cell $x$,  switches the $M_0$ and $M_1$ factors as the particle is in the right-moving state $\ket{v^x_0 v^x_1} = \ket{01}$.

Figure~\ref{figqw2} shows the projection on the left-moving single-particle sector. After stage $3$ QLGA advection, it  shows the memory repositioning by stage $3$ scattering $S_3$ in the left  cell $x-1$ (destination).  At that point, the local scattering operator $L_3$ (local part of $S_3$), acting on the left cell $x-1$,  switches the $M_0$ and $M_1$ factors as the particle is in the left-moving state $\ket{v^{x-1}_0 v^{x-1}_1} = \ket{10}$.
\begin{figure}[H] %\label{figqw4}
\includegraphics[width=0.75\textwidth]{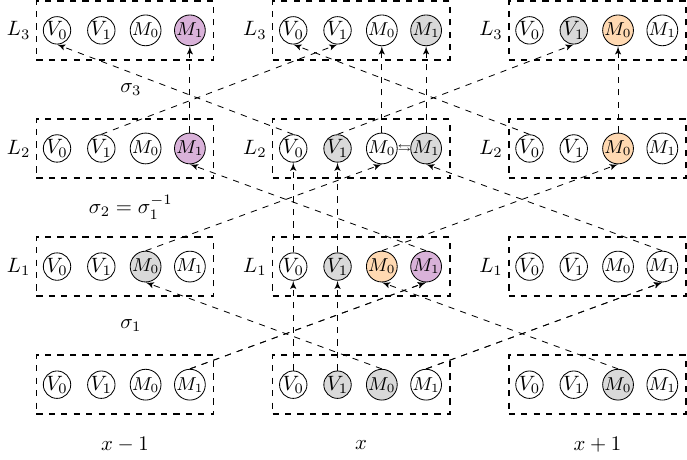}
\caption{NHQW transition as a single-particle sector of the QCA.  
From stage $1$ scattering  ($L_1$ blocks) onwards,  projection on the right-moving single-particle sector is shown.
 White: $\ket{0}$, gray: $\ket{1}$, colored: any state. Memory repositioning by  $L_2$ in the origin cell $x$   is shown by the  left-right  arrows.}
%Each QLGA stage shows projection on right-moving single-particle sector after respective scattering.    Colored  circles represent states other than $\ket{0}$ or $\ket{1}$. Center cell ($x$) memory repositioning by stage $2$ scattering $S_2$, conditionally switching $M_0$ and $M_1$,    is shown by the  left-right  arrows.}
%A QCA model of NHQW: NHQW state after each QLGA stage is shown.  Unshaded circles represent $\ket{0}$. Darker colors after scattering $L_1$  indicate NHQW scattering.  Projection on the right-moving particle sector is tracked,  illustrating the center cell memory repositioning by stage $2$ QLGA. After the second stage advection, $S_2$ interchanges $M_0$ and $M_1$ in the center cell,  to have the memory  in  valid state  before a hop to the  destination cell on the right.
\label{figqw4} %\label{figqw2}
\end{figure}

\begin{figure}[H] 
\includegraphics[width=0.75\textwidth]{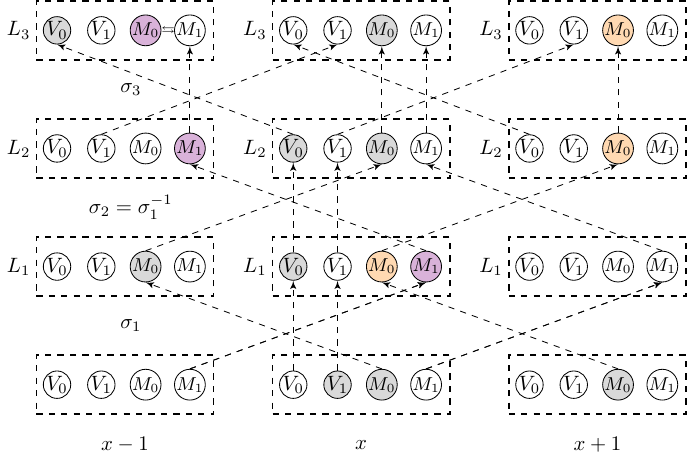}
\caption{NHQW transition as a single-particle sector of the QCA.   
From stage $1$ scattering  ($L_1$ blocks) onwards,  projection on the left-moving single-particle sector is shown.
 White: $\ket{0}$, gray: $\ket{1}$, colored: any state. Memory repositioning by  $L_3$ in the destination cell $x-1$   is shown by the  left-right  arrows.}
%A QCA model of NHQW: NHQW state after each QLGA stage is shown. Unshaded circles represent $\ket{0}$. Darker colors after scattering $L_1$  indicate NHQW scattering.  Projection on the left-moving particle sector is tracked,  illustrating the destination cell memory repositioning by stage $3$ QLGA.  After the third stage advection, $S_3$ interchanges $M_0$ and $M_1$ in the destination cell on the left, to have the memory  in  valid state.
\label{figqw2}
\end{figure}

This QCA is a concatenation of $3$ QLGA, but is not a QLGA itself when viewed at the  scale at which the dynamics are homogeneous. That means that the evolution of the QCA as a whole in eq.~\eqref{qcaevolv} is not an advection followed by scattering acting locally on each cell, even if the cell structure were to be redefined. The analysis would be similar to that for the QCA example analyzed in~\cite{bib:msqcawp}. In that paper,  a QCA that is a concatenation of  two QLGA is shown to not  be a  QLGA itself.
  
\section{Conclusion} \label{sec:conc}
In this paper, we generalize an interesting model of QW from~\cite{crt:qwwtsa1d}, that has  a memory qubit at each site to keep a history of particle's  visits and steer it. Our generalization allows a neighborhood of memory qubits in scattering interaction with the particle. We call it the neighborhood-history quantum walk (NHQW). We construct an  example NHQW on a one-dimensional lattice, with the left/right symmetric neighboring memories. This construction utilizes unentangled orthogonal bases (UOB)~\cite{lsw:lduobag} to define the scattering, hence we call it the UOB-scattering NHQW. The use of UOB helps in explicitly encoding the scattering strategy as a joint reconfiguration of the   particle velocity and neighborhood from pre-scattering  ``balanced" or ``unbalanced" neighborhood memory states and the particle velocity.  The UOB are parametrized, so the particle velocity and the memory states  undergo,  upon scattering, parametrized internal (spin) state transformations. By including  further  ``external"  interaction parameters that rotate among the elements of UOB, we gain more adjustability in scattering.  The scattering is  geared for a reasonable search strategy.  We find, through simulations,  the classical random walk and quantum walk as special cases, but also several other variations in traversal patterns. These are controllable through  appropriate parameter tuning.   Further work in higher-dimensional lattices would reveal  more fully the traversal potential of this model. We also describe a multi-particle QCA generalization of  this NHQW. That  QCA is a concatenation of $3$ QLGA, but is not a QLGA itself. Thus, it does not have a particle description at the scale at which the dynamics are homogeneous. In this respect, our NHQW differs from a host of history dependent QWs~\cite{bib:sdlhdQWqlga}, whose multi-particle generalizations are QLGA. NHQW, therefore, launches QWs into a dynamic regime with a higher degree of  complexity. In the future, we aim to analyze the search  capabilities and recurrence  properties of NHQWs. To this end, there is an extant body of work  to draw upon, concerning   measures such as   mixing time~\cite{Aharonov:2001,Venegas-Andraca2012}, hitting time~\cite{Magniez:2009,Venegas-Andraca2012}, and P\'olya Number~\cite{sjk:recpolya}. 
%%%%%%%%%%%% NEW 
NHQW have  exponential (in the number of sites) resource requirements as they require an additional qubit per lattice site. Compared to  schemes that retain memory in a finite number of additional qubits, for instance those QWs that have additional memory in the particle state (for instance, velocity) to track the particle's history,  it would require a study of  individual cases of search and algorithmic necessity to justify NHQW's use. This is  despite it having higher versatility, by definition, than the finite memory counterparts. Perhaps it is naturally a model  that  accounts for the dynamics of the system, as in noisy environments or modeling fundamental physics. Theoretical and numerical simulation of NHQW may reveal limiting behaviors and patterns that are significantly different from conventional QWs, which is expected as the dynamics are derived from a QCA. 
%%%%% ADJUSTED
Simulating QWs in the current NISQ~\cite{p:qcnisq} era of quantum computing with multitudes of coupling parameters accounting interactions among qubits of a quantum computer, a well-designed  NHQW may both facilitate estimation of the parameters, based on deviation from ideal dynamics,  and tune the walk to respond in a controlled manner, perhaps adaptively.
%%%%%%%%%%%%%%%%%%%%%%

% are embedded in the factors $\bigotimes_{k < x}\ket{m_1}_x\otimes \bigotimes_{k \geq x}\ket{m_0}_x$ of the single-particle state of the QCA After the transition (QCA evolution $R$) step, the state of the :

%\begin{figure}[H]
%\minipage{0.5\textwidth}
%  \includegraphics[width=0.85\linewidth]{ndqw1.eps}
%\endminipage\hfill
%\minipage{0.5\textwidth}%
%  \includegraphics[width=0.85\linewidth]{ndqw2.eps}
% \endminipage
%%   \caption{QW evolution over time} \label{fig:ndqw12}
%\end{figure}

%We can easily conceive of applying  this idea to higher dimensional lattices, neighborhoods that are more general, and  create generalized scatterings that are more complex.

\section*{Acknowledgements}
The author would like to acknowledge productive discussions with David Meyer and members of his research group.

\def\MR#1{\relax\ifhmode\unskip\spacefactor3000 \space\fi
  \href{http://www.ams.org/mathscinet-getitem?mr=#1}{MR#1}}

\end{document}